\documentclass[11pt]{article}

\usepackage[a4paper, margin=1in]{geometry}

\usepackage{setspace}

\usepackage{authblk}

\usepackage{tcolorbox}
\usepackage{amsmath, amssymb}

\newtcolorbox{problemBox}{
  colback=white,
  colframe=black,
  boxrule=0.5pt,
  arc=2pt,
  left=4pt,
  right=4pt,
  top=4pt,
  bottom=4pt
}

\usepackage{amsthm}
\usepackage{amsmath}
\usepackage{thmtools,thm-restate}
\usepackage{amssymb}
\usepackage{comment}
\usepackage{tikz}
\usepackage{algorithm}
\usepackage{algpseudocode}
\usepackage{hyperref}
\usepackage{outlines}
\usepackage{url}

\usepackage[numbers]{natbib}

\usepackage{xcolor, soul}

\usepackage{caption}

\newtheorem{definition}{Definition}[section]
\newtheorem{theorem}[definition]{Theorem}
\newtheorem{proposition}[definition]{Proposition}
\newtheorem{corollary}[definition]{Corollary}
\newtheorem{lemma}[definition]{Lemma}

\newtheorem{observation}[definition]{Observation}

\usepackage{orcidlink}

\usepackage{enumitem}
\setlist[itemize]{noitemsep}
\setlist[enumerate]{noitemsep}

\title{Designing Pairwise-Stable Agent Seating Arrangements}
\author[ ]{Frederik Glitzner \orcidlink{0009-0002-2815-6368}}
\affil[ ]{School of Computing Science, University of Glasgow, Glasgow G12 8QQ, UK}
\affil[ ]{\normalfont \texttt{f.glitzner.1@research.gla.ac.uk}}
\date{}

\begin{document}
\thispagestyle{empty}

\maketitle

\begin{abstract}
    Many fundamental problems in multi-agent systems involve the arrangement of agents, who have preferences over each other, on a target graph. These problems include, for example, {\sc Stable Matching}, {\sc Seat Arrangement} and {\sc Coalition Formation}. However, guaranteeing game-theoretically desirable properties such as exchange-stability or envy-freeness is difficult, as such solutions may not exist, and even if they do, they are often intractable to find, even in highly constrained settings such as path or cycle target graphs.

    In this paper, we challenge the classical setup and investigate what can be achieved when the structure of the target graph is a designable object for the central planner, rather than a fixed part of the input. We study this in the context of a natural pairwise stability criterion, which is similar to having spare seats. In particular, we introduce a highly flexible framework to efficiently design approximately optimal target graphs and associated pairwise-stable agent arrangements. Our model assumes that agents have (weak or strict) ordinal preferences over other agents. We show that classical results from stable matching theory can be extended and adapted to this much more general setting and can serve as a useful tool for navigating the trade-off between stability and computational efficiency.
    
    Our results highlight strict boundaries between tractability and intractability, and between local and global optimality. We also uncover intriguing connections to classical computational problems such as subgraph isomorphism, disjoint path partitioning, and bin-packing.
\end{abstract}

stable matching; efficient algorithm.


\section{Introduction}

\subsection{Motivation}

A wide range of problems at the intersection of discrete optimisation, economics and computer science involve arranging (self-interested) agents -- who have preferences over one another -- on a target graph structure. These agent arrangement problems include settings such as coalition formation, seat assignments, and roommate allocations, and can model various domains such as school choice \cite{schoolchoice}, ride-sharing \cite{YAO2023102775}, or collaborative environments \cite{cfg}. In these settings, agents care strongly about who they are grouped with or adjacent to. A recurring theme in this line of work is the tension between desirable game-theoretic and social choice properties -- such as stability and fairness -- and computational efficiency. Even in highly restricted settings, solutions satisfying desirable criteria may not exist, or may be computationally hard to find. For instance, in the {\sc Seat Arrangement} problem, deciding the existence of a mapping from the set of agents to the seats in a given seating graph that ensures envy-freeness (the absence of agents that would prefer the seat of a different agent over their own) or exchange stability (the absence of a pair of agents that would like to swap seats) can be NP-hard even on simple target graph topologies such as paths or cycles \cite{pepe}.

Our work is motivated by practical situations where a solution must be found despite these obstacles, and where, crucially, the structure of the target graph is not fixed, but can be designed (e.g., by a central planner). We study a general setting in which agents express weak ordinal preferences over one another, and the planner must assign them to positions in a graph in a way that avoids instability. Specifically, our main objective, which differs from other forms of stability often sought in other works on coalition formation and seat arrangement, is to find an arrangement where no pair of agents would prefer to break away from their assigned positions and form a pairing with each other outside of the target graph. We call such solutions \emph{stable} in this paper, but it is worth noting that many different notions of stability are studied in this general line of work. We introduce a new notion that could also be referred to as  \emph{pairwise-stable} or \emph{runaway-stable} to capture the fundamental idea that no two agents are incentivised to ``run away'' and self-organise outside of the given target structure instead.\footnote{The much more common objective in coalition formation problems is \emph{single-player deviation stability}, where no agent prefers to move to a different, already-existing, coalition to improve their own welfare \cite{Aziz_Savani_2016}. In seat arrangement problems, the usual stability objective is pairwise \emph{exchange-stability}, where no two agents want to swap seats \cite{pepe}.}

This new stability definition, which only requires agents to be content with at least one of their neighbours (i.e., one local connection in the target graph), turns out to offer a tractable middle ground between stronger but intractable or infeasible notions of stability and undesirable trivial solutions. Furthermore, it can be useful in settings where agents have limited visibility or communication (e.g., in the case where agents only know or interact with local neighbourhoods in the target graph), or where coordination costs prevent groups of agents from identifying mutually beneficial deviations that require more than two agents. Importantly, it also enables efficient computation in cases where stronger notions do not (under standard computational complexity assumptions).

We illustrate the relevance and applicability of this model using the following three examples:
\begin{enumerate}
    \item At a wedding or conference dinner, guests might want to be seated such that everyone has at least one person nearby that they enjoy talking to. If two guests at different tables feel ill-positioned and prefer each other to their respective best neighbours, they may leave their seats to converse elsewhere.
    \item In team formation at a hackathon, participants may accept a team as long as they have at least one strong collaborator. But if two participants prefer each other over even their best current teammates, they are incentivised to leave and form a new, smaller team.
    \item In university housing, students may accept their assigned roommates as long as there is someone in their flat they strongly enjoy living with. However, if two students who like each other feel isolated in their current flats, they might abandon their flats and move off-campus together.
\end{enumerate}

These examples are real-world scenarios where there might be several potential notions of an ``optimal'' solution, and where our stability constraint leads to more desirable outcomes than simple naive or random solutions. Our work builds on and utilises classical combinatorial results from stable matching theory -- specifically well-established theory borrowed from the well-known {\sc Stable Roommates} problem -- and shows how relevant known results can be generalised to and used in the strictly more general above-mentioned agent arrangement problems with a graph design aspect. Along the way, we discover interesting connections to classical algorithmic problems such as subgraph isomorphism, disjoint path partitioning, and bin packing.

\subsection{Our Contributions}
\label{sec:contributions}

We offer an alternative perspective to the vast majority of works in the literature that establish intractability results for various objectives even in very restricted settings. Here, we take the perspective of a ``designer'' who has some degree of freedom to design seating graphs and team sizes. For example, a wedding planner may have some freedom to choose seating layouts. We also show how our framework can be applied to other kinds of problems, such as $b$-matching with preferences.

In general, we consider a set of agents $A$ who have an ordinal preference ranking (possibly involving ties) $\succsim_i$ over all other agents (where the tuple of preference rankings is collectively represented by $\succsim$). The computational problems we consider usually involve a given set of agents and their corresponding preference system, as well as some structural objective on the solution, and the aim to compute a \emph{stable arrangement}, that is, a solution consisting of a target graph $G$ (e.g., a seating graph) and a mapping from $A$ to the vertices of $G$. We formally define the new notion of stability (in which no two agents prefer each other to their best neighbour in $G$) and our main tool, namely the concept of an $(r_1,r_2,r_3)$-bundle. Intuitively, a bundle is a compression of the preference system and contains all the necessary information that allows us to build more complex graph structures. A bundle always consists of a collection of small path components (more specifically, $r_1$ many paths of length 1, $r_2$ paths of length 2, and $r_3$ paths of length 3) and a mapping from $A$ to the vertices of the graph, with the property that the bundle itself is a stable arrangement. We show how to compute a $(r_1,r_2,r_3)$-bundle in $O(n^2)$ time (for $n$ agents) using tools from classical stable matching theory, and we establish some critical structural and computational complexity properties with regard to these bundles and their computation. This allows us to use bundles as a powerful tool: we show how to transform bundles into more desirable stable arrangements by adding edges to connect the path components in certain ways, while also keeping all neighbours in the bundle as neighbours in the final solution, and we highlight the flexibility and effectiveness of bundles as a preference system compression by applying them to a variety of stable arrangement-style problems. For example, we provide an efficient algorithm that transforms a bundle into a stable seating arrangement with a minimal number of tables. We also encounter intractability for very general problems, such as the following, which aims to decide whether the components of a bundle can be connected in such a way that neighbours are preserved and the final structure forms a (not necessarily induced) subgraph of a given graph $G$, which we refer to as the problem of arranging a bundle on a graph.

\begin{problemBox}
\textsc{Arranging Bundles} \\[4pt]
\textbf{Input:} A non-empty $(r_1,r_2,r_3)$-bundle $B$ and a graph $G=(V,E)$. \\[2pt]
\textbf{Question:} Does there exist a complete stable arrangement of $B$ on $G$?
\end{problemBox}

\begin{theorem}
    {\sc Arranging Bundles} is NP-complete even if $G$ is bipartite and has maximum degree 3.
\label{thm:arrbundlesNPcompleteintro}
\end{theorem}

However, taking the perspective of a solution designer who has the flexibility to design the target graph, rather than being given a target graph as part of the input, we can establish positive tractability results. We start with seat arrangement problems such as the following, where the goal is to design a seating graph consisting of tables (i.e., connected components) and seats (i.e., vertices) such that every table has the same number of seats, and a stable mapping from the set of agents to the vertices of the seating graph.

\begin{problemBox}
\textsc{MinTablesMinSeats} \\[4pt]
\textbf{Input:} A non-empty $(r_1,r_2,r_3)$-bundle $B$. \\[2pt]
\textbf{Output:} A seating graph $G=(V,E)$ and a complete stable arrangement $M$ such that $G$ contains a minimal number of tables, and subject to this, each table contains a minimal number of seats.
\end{problemBox}

\begin{problemBox}
\textsc{MinSeatsMinTables} \\[4pt]
\textbf{Input:} A non-empty $(r_1,r_2,r_3)$-bundle $B$. \\[2pt]
\textbf{Output:} A seating graph $G=(V,E)$ consisting of connected components and a complete stable arrangement $M$ such that $G$ contains a minimal number of seats at each table, and subject to this, $G$ contains a minimal number of tables.
\end{problemBox}

Using our framework, we will establish the following tractability results.

\begin{theorem}
    {\sc MinTablesMinSeats} and {\sc MinSeatsMinTables} can be solved in $O((n^2+s_*^3)n^2)$ and $O(s_*^3n^2)$ time, respectively, where $n$ is the number of agents and $s_*$ is the minimal number of seats per table. Furthermore, in the solutions, all tables are paths, and any set of edges can be added to transform the tables into other graph structures (e.g., cycles, grids, or cliques). Finally, for any table size of at least 3, there always exists a stable seating arrangement.
\label{thm:mintablesseatsintro}
\end{theorem}

Next, we consider team (or coalition) formation problems. Notice that in the language of graphs, a team can simply be considered as a clique: a collection of vertices which can be occupied by team members such that every team member is connected to every other team member. This allows us to unify team formation and seat arrangement problems in a common framework. We consider the analogous problems \textsc{MinTeamsMinSize} and \textsc{MinSizeMinTeams} to \textsc{MinTablesMinSeats} and \textsc{MinSeatsMinTables} with the only modification being that a solution must consist only of clique components, each consisting of the same number of vertices. The minimisation objectives now are in terms of team sizes and the number of teams, of course. We will establish the following.

\begin{theorem}
{\sc MinTeamsMinSize} and {\sc MinSizeMinTeams} can be solved in $O((n^2+{s_*}^3)n^2)$ and $O({s_*}^3n^2)$ time, respectively, where $n$ is the number of agents and $s_*$ is the minimal team size. Furthermore, whenever the team size is chosen to be at least 3, there always exists a stable team formation.
\label{thm:minteamssizeintro}    
\end{theorem}

Lastly, we illustrate how our techniques can be applied to other types of problems, taking a variant of stable $b$-matching as an example. Here, we consider stable $b$-matching as a matching problem where every agent $a_i$ has a capacity $b_i$ representing an upper bound on the number of matches that $a_i$ can be part of. We will call the following problem \textsc{Best-Stable $b$-Matching} to emphasize our stability notion where no two agents strictly prefer each other to their best neighbour and to differentiate this problem from more ``classical'' versions of the stable $b$-matching problem (e.g., as studied in \cite{Fleiner08}) in which no two unmatched agents may prefer each other to the \emph{worst} of their respective partners, which also means that a stable matching does not always exist. Note that in our problem definition, we assume that every agent has a capacity of at least 2 in order to ensure that a best-stable $b$-matching always exists (however, this restriction does not ensure the existence of a stable matching in the classical setting, as \citet{glitzner2025unsolvabilitymanytomanynonbipartitestable} showed). A formal problem definition of the problem follows below.

\begin{problemBox}
\textsc{Best-Stable $b$-Matching} \\[4pt]
\textbf{Input:} A non-empty $(r_1,r_2,r_3)$-bundle $B$ and a capacity function $b : A\rightarrow\mathbb{Z}^{\geq 2}$. \\[2pt]
\textbf{Output:} A maximal $b$-matching $M$ of the agents in the bundle such that no two agents prefer each other to their best partner.
\end{problemBox}

Notice the \emph{maximal} condition: this is important to mitigate \emph{wastefulness}, i.e., to avoid that two mutually acceptable agents that still have free capacity are not matched to each other. We will present the following positive tractability result for this problem.

\begin{theorem}
    {\sc Best-Stable $b$-Matching} can be solved in $O(n^2)$ time (where $n$ is the number of agents) using Algorithm \ref{alg:bmatching}.
    \label{thm:bestbmatchingintro}
\end{theorem}

\subsection{Related Work}

This paper relates to three areas: stable matching, seat arrangement, and coalition formation.

\textbf{Stable matching} 
has been studied extensively, starting with the seminal paper by \citet{gale_shapley}, who introduced the {\sc Stable Marriage}, {\sc Hospitals/Residents}, and {\sc Stable Roommates} ({\sc sr}) problems. It is well-known that the first two problems always admit a stable matching, while the third problem may not. \citet{irving_sr_structure, irving_sr, gusfield_sr_structure,gusfield89} published an extensive collection of structural results on stable matching and presented an algorithm to find a stable matching efficiently in the {\sc sr} setting with strict preferences, if one exists. For weak preferences, there are three different stability definitions for {\sc sr} \cite{irving_srt}. Most importantly for this paper, deciding the existence of a \emph{weakly stable matching} (a matching in which no two agents strictly prefer each other to their partners, either of which may be none) is known to be NP-complete in {\sc sr} \cite{ronn_srt}. For other versions of stability, positive tractability results are known \cite{irving_srt, ScottPhD,kunysz_srt_strongstable}. In the presence of capacitated agents (relevant to Section \ref{sec:matching}), so-called \emph{stable b-matching} problems are also known to be solvable in polynomial time for strict preferences \cite{fleiner03,Fleiner08,glitzner2025unsolvabilitymanytomanynonbipartitestable}, but negative results from {\sc sr} carry over by generalisation. We refer to \citet{matchup} for an extensive overview of the algorithmics of stable matching and related problems.

\textbf{Seat arrangement} problems have only been studied fairly recently, starting with a paper by \citet{Bodlaender2020Hedonic}. In their model, they assumed that agents have utilities over other agents, and adopted an additive utility function over neighbours in the seating graph to aggregate agents' utilities. They showed, for example, that deciding the existence of an envy-free allocation of agents on the seating graph is computationally intractable even when preferences are symmetric, and every graph component consists of at most three vertices. Finding a welfare-maximising allocation is also NP-hard, even when preferences are symmetric and binary. Later, \citet{ceylan,Ceylan2023Optimal} considered related problems from a parametrised complexity point of view, still encountering strong intractability barriers. \citet{Massand} studied a more general model where agents also provide utilities over seats and found that, under symmetric preferences, an exchange-stable allocation can be found efficiently. \citet{Berriaud,pepe,Wilczynski2023Ordinal} took a different approach and studied strong restrictions on the seating graphs, among other settings. Still, most problems remain NP-hard, even on very simple graph structures such as collections of paths or cycles. 

\textbf{Coalition formation}, specifically in the context of \emph{hedonic preferences} (where agents only care about their own outcome), is very well-studied \cite{Aziz_Savani_2016}. However, these stability concepts are usually restricted to individual stability, where individual agents do not benefit from moving between partitions. More similar to our stability criterion, where agents only consider the best partner/neighbour, \citet{Cechlarova2001Stability,Cechlarova2003Computational,CECHLAROVA2004333} studied stability concepts that rely solely on the best or worst agent in the partition. 

\citet{roommateroom} recently considered a model that sits between all three of these problems, aiming to assign students with preferences to rooms of fixed capacity. Similarly, \citet{Cseh2019Pareto} studied Pareto-optimality in a related setting of roommate assignments for rooms of different sizes. However, our model differs from theirs in multiple ways; most importantly by allowing a degree of flexibility when designing the arrangement/seating graph rather than starting with a fixed set of rooms, and in the stability condition we consider. 

\subsection{Structure of the Paper}

Section \ref{sec:background} provides formal definitions, a primer on crucial structural tools from the area of non-bipartite stable matching, and a brief comparison of our stability notion to other stability and optimality criteria. Section \ref{sec:construction} develops our main framework, formally defines the notion of a bundle, and shows how to find one efficiently and what properties it exhibits. Section \ref{sec:applications} highlights a number of applications of our framework to seat arrangement, team formation, and $b$-matching problems. Section \ref{sec:conclusion} provides a brief summary of the results and some future research directions.

\section{Formal Definitions and Preliminary Results}
\label{sec:background}

\subsection{Formal Definitions}
\label{sec:formal}

We will consider the following multi-agent model.

\begin{definition}[Preference System]
    Let $A=\{a_1,\dots,a_n\}$ be a set of $n$ \emph{agents} and let each agent $a_i\in A$ have a weak \emph{preference order} $\succsim_i$ over all remaining agents $A\setminus\{a_i\}$ (i.e., we assume \emph{complete} preferences), then $I=(A,\succsim)$ is a \emph{preference system}. If $a_j\succ_i a_k$ then we say that $a_i\in A$ (strictly) \emph{prefers} agents $a_j$ to $a_k$, and if $a_j \sim_i a_k$ then we say that $a_i$ is \emph{indifferent} between them, which is also the case when $a_j=a_k$.
\end{definition}

Notice that the definition above is coherent with instances of the {\sc Stable Roommates with Ties} problem ({\sc srt}), where the goal is to match a non-bipartite set of agents into unordered pairs such that no agent is in more than one pair and the matching is stable (according to one of three different notions of stability \cite{irving_srt}). Specifically, a matching is \emph{weakly stable} if no two agents strictly prefer each other to their partners (or remain unmatched). 

We adopt the following abstract problem setting for agents with weak ordinal preferences.

\begin{definition}[Arrangement Problem]
    Let $G=(V,E)$ be an undirected simple graph and let $I=(A,\succsim)$ be a preference system, then we call an injective mapping $M$ of a subset of the agents $A'\subseteq A$ to the vertices of $G$ (i.e., $M:A'\rightarrow V$) an \emph{arrangement}. We refer to $M$ as \emph{complete} if $A'=A$ (i.e., all agents are in the arrangement) and we refer to it as \emph{bijective} if it is also surjective (i.e., all vertices are covered). Finally, we refer to $N(a_i)\subset A'$ as \emph{neighbours} of $a_i$ in $M$, i.e., $N(a_i)=\{a_j\in A' \;\vert\; \{M(a_i),M(a_j)\}\in E\}$.
\end{definition}

We will focus on two kinds of arrangement problems: either we will be given the graph $G$ and ask whether an arrangement with certain properties exists (and potentially find one), or we want to construct a suitable graph $G$ such that an arrangement with certain properties is guaranteed to exist (and we can find it). Usually, we are interested only in arrangements that are either complete or bijective (leaving out agents and leaving seats free at the same time is rarely good). The stability notion we newly introduce and adopt for this work is the following.

\begin{definition}[Stability]
    Let $I=(A,\succsim)$ be an preference system, let $G=(V,E)$ be a graph and let $M$ be an arrangement of $A'\subseteq A$ to $V$. Then we denote the \emph{best neighbour} of an agent $a_i\in A'$ (i.e., the best agent $a_j$ such that $a_j\in N(a_i)$, assuming an arbitrary tie break) by $b(a_i)$. If $N(a_i)=\varnothing$, then $b(a_i)$ is null. Now we refer to two distinct agents $a_i,a_j\in A'$ as \emph{blocking} (or as a \emph{blocking pair}) if $a_i$ (strictly) prefers $a_j$ to $b(a_i)$ and vice versa. If there do not exist any blocking pairs, we call $M$ \emph{stable}. 
\end{definition}

Notice the subtle difference in notation between the capacity $b_i$ of an agent $a_i$ in the $b$-matching problem and the best neighbour of an agent $b(a_i)$. Also, we want to reiterate that this stability notion is motivated as follows: we want to ensure that no two agents have an incentive to leave the market, or, in a more practical setting of seat arrangements, that no two agents leave the room together in order to be better off.

\subsection{Stable Partition Background}
\label{sec:spbackground}

Most of our main results are based on observations about the stable partition structure, which was introduced by \citet{tan91_1} as a succinct certificate for the non-existence of a stable matching in the {\sc Stable Roommates} problem (later referred to as {\sc sr}, which shares the same properties as {\sc srt} but assumes strict preference lists). \citet{tan91_1,tan91_2} showed that stable partitions always exist, satisfy various desirable properties, and gave an algorithm (referred to as \emph{Tan's algorithm}) that finds one for a problem instance with $n$ agents in $O(n^2)$ time. 

\begin{definition}[Stable Partition]
\label{def:sp}
    Let $I=(A,\succ)$ be an {\sc sr} instance. Then $\Pi$ is a \emph{stable partition} of $I$ if it is a permutation of $A$ and  
    \begin{enumerate}
        \item[(T1)] $\forall a_i \in A$ we have $\Pi(a_i) \succsim_i \Pi^{-1}(a_i)$, and
        \item[(T2)] $\nexists \;a_i, a_j \in A$ with $a_i\neq a_j$ such that $a_j \succ_i \Pi^{-1}(a_i)$ and $a_i \succ_j \Pi^{-1}(a_j)$,
    \end{enumerate}
    where $\Pi(a_i)$ is the successor and $\Pi^{-1}(a_i)$ is the predecessor of $a_i$ in $\Pi$.
\end{definition}

The author also provided the following characterization of stable partitions.

\begin{theorem}[\cite{tan91_1, tan91_2}]
\label{thm:tan}
    The following properties hold for any {\sc sr} instance $I$ with $n$ agents.
    \begin{itemize}
        \item Any two stable partitions $\Pi_a, \Pi_b$ of $I$ contain exactly the same cycles of odd length $\mathcal{O}_I$.
        \item $I$ admits a stable matching if and only if no stable partition of $I$ contains an odd cycle.
        \item Any stable partition can be broken down into a reduced stable partition in $O(n)$ time by replacing all cycles of even length by collections of transpositions.
    \end{itemize}
\end{theorem}

Notice that if an {\sc sr} instance $I$ admits a stable matching $\{\{a_{i_1}, a_{i_2}\}, \dots, \{a_{i_{2k-1}}, a_{i_{2k}}\}\}$ consisting of $k$ pairs, then it can be denoted interchangeably with its induced collection of transpositions $\Pi=(a_{i_1}$ $a_{i_2})\dots (a_{i_{2k-1}}$ $a_{i_{2k}})$, which is a stable partition of $I$. Furthermore, if the stable partition corresponds to a matching (i.e., the collection of odd cycles of odd length at least 3, also denoted by $\mathcal{O}_I^{\geq 3}$, is empty), then it is stable. 

\begin{observation}
    If the ties of an {\sc srt} instance $I$ are broken arbitrarily, we get an {\sc sr} instance $I'$. Any stable partition $\Pi$ of $I'$ is also a stable partition of $I$, i.e., $\Pi$ remains a permutation of the agents and satisfies T1 and T2 with respect to the preference system $I$.
\end{observation}

Now, in the case of an {\sc srt} instance $I$, if $I$ admits a weakly stable matching consisting of $k$ pairs, then it can be denoted interchangeably with its induced collection of transpositions (and one agent in a cycle of length 1 if the number of agents is odd), which is a stable partition of $I$. Furthermore, if a stable partition $\Pi$ corresponds to a matching (potentially with one agent in a length 1 cycle), then the set of pairs derived from the transpositions of $\Pi$ is a weakly stable matching of $I$. Although we can compute a stable partition in polynomial time, this does not contradict the NP-hardness of deciding whether $I$ admits a weakly stable matching \cite{ronn_srt}. This is because Theorem \ref{thm:tan} does not carry over to weak preferences and so, depending on the tie break (and there can be exponentially many ways to break the ties), different stable partitions might have different odd cycles, regardless of whether $I$ admits a weakly stable matching.

\subsection{Compatibility with Other Concepts}

Our stability notion is fundamentally different from other optimality criteria that we highlighted in the related work section. For example, an exchange-stable arrangement is not necessarily stable, because there might exist two agents that prefer each other to their best partners, but none of their neighbours is willing to exchange positions. Similarly, stability does not imply exchange-stability either because two agents could both be better off by exchanging their positions, as long as neither of their new neighbours experiences an improvement with regard to their best neighbour. Also, stability does not imply envy-freeness because an agent $a_i$ might envy another agent $a_j$ because they are neighbouring with someone that $a_i$ prefers to their best neighbour in the arrangement. Envy-freeness, however, \emph{does} imply stability because if no agent envies the position of any other agent in the arrangement, then they must already be neighbouring their most preferred agent (or one of their most preferred agents if there are ties in the preference lists) in their preference ranking.

Now we will highlight that our stability notion, which is fundamentally focused on the non-existence of deviation incentives between \emph{pairs} of agents, is a reasonable notion to consider even when \emph{coalitions} (or, more formally, arbitrarily sized subsets of agents) may still have an incentive to exit the market together, i.e., when they do not participate in the computed stable arrangement but collectively leave the scheme they all have someone in the coalition that is strictly better than their best neighbour in the stable arrangement. Notice that if such a coalition $C$ exists, it must be the case that $\vert C\vert>2$ and that no pair of agents in $C$ is better off because of each other (otherwise they would violate stability). In practical terms, this means that unless agents are fully aware of each other's complete preferences, in order for all agents in $C$ to identify the opportunity to exit the market together, every agent in $C$ must necessarily communicate and agree to cooperate with agents that they prefer \emph{less} than their current best neighbour. Especially in settings such as our motivating examples of large team projects or conference dinners, this means that agents must communicate, coordinate, and cooperate across teams or tables with agents that they have no pairwise incentive to exit the market with to identify these larger coalitions. As these coalitions may be of linear length in the number of agents even when they exist, this is generally infeasible in practice and therefore renders our pairwise stability notion reasonably robust and well-motivated.

\section{Constructing Stable Subgraphs}
\label{sec:construction}

\subsection{Introducing Bundles}
\label{sec:introbundles}

We will use a series of examples to illustrate the establishment of our key properties. Our goal is to show that we can divide our agents into small components (this can be thought of as ``cutting up the preference graph'') to later ``glue'' the components back together in a meaningful way. As a running example, consider the preference system $I$ (or similarly, a preference system with ties broken arbitrarily) with $n=15$ agents that admits the stable partition 
$$\Pi=(a_1\;a_2\;a_3\:a_4\;a_5)(a_6\;a_7\;a_8\;a_9)(a_{10}\;a_{11}\;a_{12})(a_{13}\;a_{14})(a_{15}),$$
illustrated in Figure \ref{fig:stablepartition} (with partially shown strict preferences in blue colour, but note that other choices of preferences are also possible to achieve the same stable partition). 

\begin{figure}[!htb]
    \centering
    \includegraphics[width=12cm]{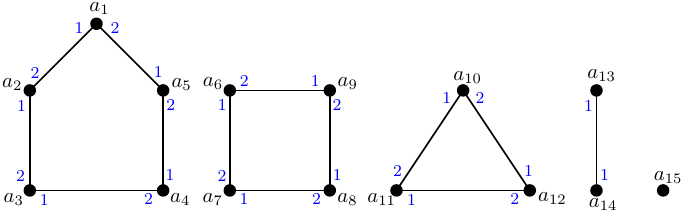}
    \caption{An {\sc sr} instance and a stable partition.}
    \label{fig:stablepartition}
\end{figure}

Now it is natural to consider a stable partition as an agent arrangement -- we simply consider the (bijective) identity mapping from the agent set to their place on the stable partition graph as the arrangement. We will refer to this as a \emph{stable partition arrangement}. Notice that this construction always exists and is stable.

\begin{restatable}{lemma}{sprs}
\label{lemma:sprs}
    A stable partition arrangement $M$ is a stable arrangement.
\end{restatable}
\begin{proof}
    Suppose that the arrangement is not stable, i.e., that there exists a blocking pair consisting of agents $a_i$ and $a_j$ that strictly prefer each other to $b(a_i)$ and $b(a_j)$, respectively. Then they cannot be neighbours in the arrangement because for all neighbours $n\in N(a_i)$ of $a_i$ in $M$, $b(a_i)\succsim n$ (and similarly for the neighbours of $a_j$). However, then $a_j\succ_ib(a_i)\succsim_i \Pi^{-1}(a_i)$ and $a_i\succ_jb(a_j)\succsim_j \Pi^{-1}(a_j)$ by construction, contradicting T2 of Definition \ref{def:sp}.
\end{proof}

Note that $\Pi$ is not a unique stable partition of $I$, as we can replace the cycle $(a_6\;a_7\;a_8\:a_9)$ by either $(a_6\;a_7)(a_8\:a_9)$ or $(a_6\;a_9)(a_7\;a_8)$ to arrive at a reduced stable partition. The former case is illustrated in Figure \ref{fig:reducedpartition}.

\begin{figure}[!htb]
    \centering
    \includegraphics[width=12cm]{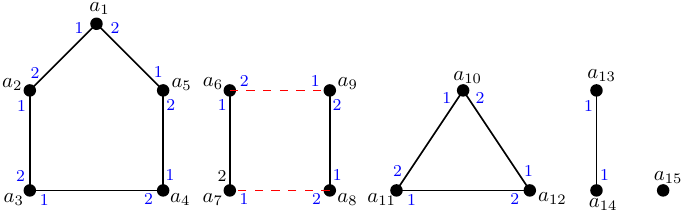}
    \caption{An {\sc sr} instance and a reduced stable partition}
\label{fig:reducedpartition}
\end{figure}

Now that this is a reduced stable partition, we cannot make any of the graph components smaller (i.e., involving fewer agents) or remove any edges from the arrangement while keeping it a stable partition. However, if we are only interested in our new stability notion and can discard the stable partition requirements, we can also open up the loops to gain ``stable paths'' instead. Figure \ref{fig:cutreducedpartition} shows such an arrangement -- we simply removed edges $\{a_1,a_5\}$ and $\{a_{10},a_{12}\}$, but other choices would be possible.

\begin{figure}[!htb]
    \centering
    \includegraphics[width=12cm]{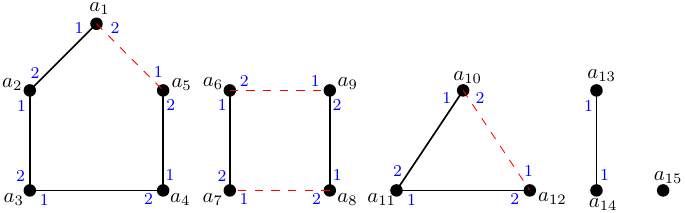}
\caption{An arrangement based on a cut-up reduced stable partition}
\label{fig:cutreducedpartition}
\end{figure}

Consistent with the fundamentals of graph theory, we will refer to a path and a cycle involving $k$ vertices (or agents) as a $P_k$ and as a $C_k$, respectively.

\begin{restatable}{lemma}{cutcycles}
\label{lemma:cutcycles}
    For any $k\geq 3$, turning a $C_k$ into a $P_k$ in a stable partition arrangement by eliminating any one edge from the cycle gives another stable arrangement.
\end{restatable}
\begin{proof}
    Recall that the stable partition arrangement $M$ based on a stable partition $\Pi$ is stable by Lemma \ref{lemma:sprs}, i.e., it does not admit a blocking pair. Now suppose that after eliminating edge $\{a_i,a_j\}$ from $C_k$, the arrangement $M'$ is not stable anymore. Naturally, at least one of $a_i,a_j$ must be in a blocking pair admitted by $M'$. Without loss of generality, suppose that $\{a_i,a_r\}$ is a blocking pair. Notice that as $k\geq 3$, $a_i$ is still connected in $M'$ to one of its neighbours from $M$, therefore, by construction of the stable partition arrangement, $b_{M'}(a_i)\succsim_i \Pi^{-1}(a_i)$, where $b_{M'}(a_i)$ is the best neighbour of $a_i$ in $M'$. Similarly, $a_r$ is also still connected in $M'$ to one of its neighbours from $M$, therefore $b_{M'}(a_r)\succsim_r \Pi^{-1}(a_r)$. However, by assumption, $a_i,a_r$ block in $M'$, that is, $a_r\succ_i b_{M'}(a_i)$ and $a_i\succ_r b_{M'}(a_r)$, therefore $a_r\succ_i\Pi^{-1}(a_i)$ and $a_i\succ_r\Pi^{-1}(a_r)$, which contradicts property T2 of stable partition $\Pi$.
\end{proof}

Knowing this, and by the fact that a stable partition always exists, we can always find a stable arrangement on some collection of paths. Now, notice that if we want to minimise the sizes of the graph components further, we can also split off paths involving 2 agents (i.e., $P_2$ subgraphs) as long as we do not increase the number of isolated vertices (leaving a single agent without any agent they prefer at least as much as their predecessor, which likely introduces blocking pairs). One such choice is shown in Figure \ref{fig:cutsplitreducedpartition}, where edge $\{a_2,a_3\}$ is eliminated. Notice that we cannot split off any more edges at this point without creating isolated vertices.

\begin{figure}[!htb]
    \centering
    \includegraphics[width=12cm]{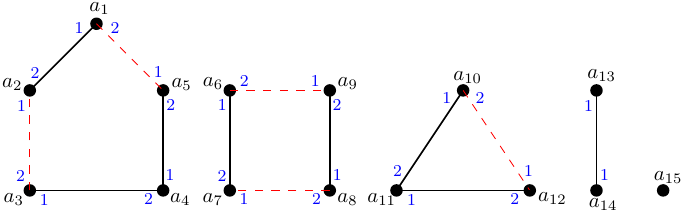}
\caption{An arrangement with more components based on a cut-up reduced stable partition}
\label{fig:cutsplitreducedpartition}
\end{figure}

\begin{restatable}{lemma}{splitpaths}
\label{lemma:splitpaths}
    For any $k\geq 5$, turning a $P_k$ into a $P_{k-2}$ and a $P_2$ in the construction above by eliminating an edge maintains stability.
\end{restatable}
\begin{proof}
    This follows by a similar argument as in Lemma \ref{lemma:cutcycles} by the fact that every agent will still be connected to someone at least as good as their predecessor.
\end{proof}

This gives us the following initial result for our proposed notion of stability.

\begin{proposition}
\label{prop:p3}
    Any preference system instance admits a stable arrangement on a collection of (sufficiently many) $P_3$ components.
\end{proposition}

Notice that this may require more vertices than agents, with some vertices unassigned. To be more specific than this, we will characterise instances in terms of bundles (collections) of $P_1$ (isolated vertices), $P_2$ and $P_3$ components that together permit a stable arrangement.

\begin{definition}
    Let $I=(A,\succsim)$ be a preference system. Then $I$ admits a \emph{$(r_1,r_2,r_3)$-bundle} $B=(G_B,M_B)$ if $G_B$ is a graph consisting of $r_1$ isolated vertices, $r_2$ $P_2$ components and $r_3$ $P_3$ components, and $M_B$ is a bijective stable arrangement of the agents $A$ on $G_B$.
\end{definition}

Reorganising our arrangement in Figure \ref{fig:cutsplitreducedpartition}, we get a bundle as shown in Figure \ref{fig:bundle}. 

\begin{figure}[!htb]
    \centering
    \includegraphics[width=12cm]{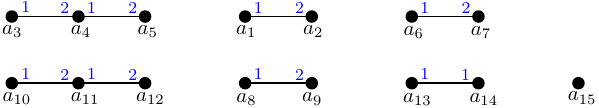}
\caption{A bundle based on a cut-up reduced stable partition}
\label{fig:bundle}
\end{figure}

\subsection{Computation and Properties of Bundles}

The results summarised in Sections \ref{sec:spbackground} and \ref{sec:introbundles} allow us to find $(r_1,r_2,r_3)$-bundles efficiently for any preference system. Algorithm \ref{alg:klmbundle} gives a detailed procedure that formalises the steps outlined above: ties are broken arbitrarily, then a stable partition $\Pi$ is computed, even-length cycles are decomposed into transpositions, fixed-points are added as $P_1$s, transpositions are added as $P_2$s, and odd-length cycles are decomposed and added as $P_2$s and $P_3$s.

\begin{algorithm}[!htb]
{\small
\renewcommand{\algorithmicrequire}{\textbf{Input:}}
\renewcommand{\algorithmicensure}{\textbf{Output:}}

    \begin{algorithmic}[1]

    \Require{$I=(A,\succsim)$ : a preference system}
    \Ensure{$M_B$ : an arrangement; $G_B,(r_1,r_2,r_3)$ : a bundle graph $G_B$ and its properties}

    \State $M_B\gets (A \rightarrow \varnothing)$
    \State $G_B\gets \varnothing$
    \State $r_1,r_2,r_3\gets 0$
    
    \State break ties in $\succsim$ arbitrarily (if any) to give $I'=(A,\succ)$
    \State compute a reduced stable partition $\Pi$ of $I'$
    
    \For{each cycle $C=(a_{i_1}\;a_{i_2}\;\dots a_{i_q})\in \Pi$}
        \If{$q=1$}
            \State $r_1\gets r_1+1$
            \State $G_B$.add$(P_1[x])$ \Comment{Add a $P_1$ with new vertex $x$ to $G_B$}
            \State $M_B(a_{i_1})\leftarrow x$
        \ElsIf{$q=2$}
            \State $r_2\gets r_2+1$
            \State $G_B$.add$(P_2[x,y])$ \Comment{Add a $P_2$ with new vertices $x,y$ to $G_B$}
            \State $M_B(a_{i_1})\leftarrow x$
            \State $M_B(a_{i_2})\leftarrow y$
        \Else
            \State $r_3\gets r_3+1$
            \State $G_B$.add$(P_3[x,y,z])$ \Comment{Add a $P_3$ with new vertices $x,y,z$ to $G_B$}
            \State $M_B(a_{i_1})\leftarrow x$
            \State $M_B(a_{i_2})\leftarrow y$
            \State $M_B(a_{i_3})\leftarrow z$
            \If{$q>3$} \Comment{Split $P_2$ components from odd-length cycles of length $>3$}
                \For{$0\leq r<\tfrac{q-3}{2}$}
                    \State $r_2\gets r_2+1$
                    \State $G_B$.add$(P_2[x',y'])$ \Comment{Add a $P_2$ with new vertices $x',y'$ to $G_B$}
                    \State $M_B(a_{i_{2r+4}})\leftarrow x'$
                    \State $M_B(a_{i_{2r+5}})\leftarrow y'$
                \EndFor
            \EndIf
        \EndIf
    \EndFor
    
    \State\Return{$M_B,G_B,(r_1,r_2,r_3)$}

    \end{algorithmic}
    \caption{Constructs a $(r_1,r_2,r_3)$-bundle $B$}
    \label{alg:klmbundle}}
\end{algorithm}

\begin{theorem}
    \label{thm:klmbundle}
    Let $I=(A,\succsim)$ be a preference system with $n$ agents, let $I'$ be the preference system with strict preferences (under any tie break), and let $\Pi$ be a reduced stable partition of $I'$ with odd-length cycles $\mathcal{O}_{I'}$. Then $I$ admits a $(r_1,r_2,r_3)$-bundle with 
    \begin{align*}
        &r_1=\vert\mathcal{O}_{I'}\vert-\vert\mathcal{O}_{I'}^{\geq 3}\vert \in\{0,1\};\\
        &r_2=\tfrac{n-\vert\mathcal{O}_{I'}\vert}{2}-\vert\mathcal{O}_{I'}^{\geq 3}\vert \in [0,\tfrac{n}{2}]\subset\mathbb{Z};\\
        &r_3=\vert\mathcal{O}_{I'}^{\geq 3}\vert \in [0,\tfrac{n}{3}]\subset\mathbb{Z}.
    \end{align*}
    We can find such a bundle and a stable arrangement using Algorithm \ref{alg:klmbundle} in $O(n^2)$ time. 
\end{theorem}
\begin{proof}
    We will argue this through the constructive method given in Algorithm \ref{alg:klmbundle}. Notice that we start with a reduced stable partition, i.e., all cycles are either of length 1, of length 2, or of odd length at least 3. This can be achieved, for example, using Tan's algorithm \cite{tan91_1} and a subsequent decomposition of even-length cycles into collections of transpositions. Notice also that the cases ``\textbf{if} $q=1$'' and ``\textbf{else if} $q=2$'' correspond precisely to adding cycles of length 1 as $P_1$ and cycles of length 2 $P_2$ to the bundle and mapping them as in the stable partition arrangement in Lemma \ref{lemma:sprs}. Furthermore, the third case ``\textbf{else}'' precisely splits an odd-length cycle into one $P_3$ and potentially multiple $P_2$ (as many as required by the length of the cycle to include all agents), the stability of which we argued in Lemmas \ref{lemma:cutcycles} and \ref{lemma:splitpaths}. Thus, the bundle $B$ and corresponding arrangement $M$ computed in Algorithm \ref{alg:klmbundle} are stable.

    Recall that $\mathcal{O}_{I'}$ are the odd-length cycles of $I'$ and that $\mathcal{O}_{I'}^{\geq 3}$ is the subset of all cycles of odd length at least 3. Clearly $r_1=\vert\mathcal{O}_{I'}\vert-\vert\mathcal{O}_{I'}^{\geq 3}\vert$. Furthermore, $r_2$ is the number of $P_2$, which is precisely half the number of agents that are not in $P_1$ and not in $P_3$. We already argued that there are $\vert\mathcal{O}_{I'}\vert-\vert\mathcal{O}_{I'}^{\geq 3}\vert$ agents in $P_1$ and we described that the algorithm adds one $P_3$ for every cycle of odd length at least 3, i.e., $r_3=\vert\mathcal{O}_{I'}^{\geq 3}\vert$. Thus, $r_2=\tfrac{n-(r_1+3r_3)}{2}=\tfrac{n-(\vert\mathcal{O}_{I'}\vert-\vert\mathcal{O}_{I'}^{\geq 3}\vert+3\vert\mathcal{O}_{I'}^{\geq 3}\vert)}{2}=\tfrac{n-\vert\mathcal{O}_{I'}\vert}{2}-\vert\mathcal{O}_{I'}^{\geq 3}\vert$.

    Now notice that any stable partition will contain at most one cycle of length 1 due to our assumption of complete preferences (any two agents in cycles of length 1 would strictly prefer each other to their predecessors in the stable partition, contradicting stability), thus $r_1\in\{0,1\}$. The upper bounds for $r_2$ and $r_3$ follow trivially from the number of agents each component contains.

    Regarding the time complexity of Algorithm \ref{alg:klmbundle}, notice that we can break ties arbitrarily in $O(n^2)$ time by iterating through the preference lists once (we have $n$ agents and each agent has a preference list of length $n-1$). Recall that Tan's algorithm runs in $O(n^2)$ time \cite{tan91_1} and that even-length cycles can be decomposed in $O(n)$ time (see Theorem \ref{thm:tan}). Finally, notice that $\Pi$ contains $O(n)$ cycles and each cycle is of length $O(n)$, and hence the outer for loop runs for $O(n)$ iterations and the inner for loop for $O(n^2)$ iterations over the entire algorithm's execution. Furthermore, with an appropriate choice of data structures, the set and function modifications can be performed in constant time, giving the time complexity $O(n^2)$.
\end{proof}

Notice that even for strict preferences, $(r_1,r_2,r_3)$-bundles need not be unique. 

\begin{proposition}
    Even when preferences are strict, a preference system may admit a $(r_1,r_2,r_3)$-bundle and a $(r_1',r_2',r_3')$-bundle with $r_3\neq r_3'$.
\end{proposition}
\begin{proof}
    Consider a preference system with three agents and preferences $a_1:a_2\;a_3$, $a_2:a_1\;a_3$ and $a_3:a_1\;a_2$. Then we could either arrange $a_1,a_2$ on a $P_2$ and single out $a_3$ on a $P_1$, or arrange all three agents on a $P_3$ with $a_1$ on the middle vertex. In both arrangements, no two agents prefer each other to their best neighbour.
\end{proof}

However, it will become clear later that bundles with components as small as possible are desirable. We can show that these bundles computed using Algorithm \ref{alg:klmbundle} are close to optimal for instances with strict preference lists (when the goal is to minimise the number of $P_3$ components in the bundle).

\begin{theorem}
    \label{thm:strict}
    Let $I=(A,\succsim)$ be a preference system with strict preferences. If $I$ is solvable, then the $(\vert\mathcal{O}_I\vert-\vert\mathcal{O}_I^{\geq 3}\vert,\frac{n-\vert\mathcal{O}_I\vert}{2}-\vert\mathcal{O}_I^{\geq 3}\vert,\vert\mathcal{O}_I^{\geq 3}\vert)$-bundle $B$ computed using Algorithm \ref{alg:klmbundle} has a minimum number of $P_3$ components. If $I$ is unsolvable, for any $(r_1,r_2,r_3)$ bundle $B'$, $r_3\geq \vert\mathcal{O}_I^{\geq 3}\vert-1$.
\end{theorem}
\begin{proof}
    If $I$ is solvable, then $I$ admits a stable matching, and no reduced stable partition of $I$ contains any cycles of length longer than 2. Therefore, $\vert\mathcal{O}_I^{\geq 3}\vert=0$ and clearly there can be no bundle with fewer than 0 $P_3$ components.

    Now consider the case when $I$ is unsolvable. To establish our result, we will consider internally stable matchings, i.e., matchings $S$ of a subset of the agents in $I$ such that no two agents that are matched in $S$ prefer each other to their respective partners in $S$. Now we can show that $B'$ implies the existence of an internally stable matching of size $\tfrac{n-(r_1+r_3)}{2}$. Consider the complete stable arrangement $M$ of the agents on $B'$ (which must exist by definition of a bundle) and consider the set $S$ consisting of all pairs of agents that are mapped to the same $P_2$ by $M$ and all pairs of agents $a_i,a_j$ such that 
    \[\begin{cases}
        \text{there exists agent $a_k$ such that $M$ maps $a_i,a_j,a_k$ to the same $P_3$, and}\\
        \text{$a_i,a_j$ are connected in the $P_3$, and}\\
        \text{$a_j,a_k$ are connected in the $P_3$, and}\\
        \text{$a_i\succ_j a_k$.}
    \end{cases}\]
    Clearly, $S$ is a matching because no agent can be mapped to two different vertices in $B'$. $S$ is also internally stable because, by the stability of $M$, no two agents prefer each other to their best neighbour, and notice that any agent matched in $S$ is partnered with their best neighbour in $S$. Thus, no two agents matched in $S$ can prefer each other to their partner in $S$, so $S$ is internally stable. $\vert S\vert =\tfrac{n-(r_1+r_3)}{2}$ follows by construction because we do not match agents in a $P_1$, and we leave out exactly one agent from each $P_3$.
    
    \citet{tan91_2} showed that a maximum-size internally stable matching (also referred to as a maximum stable matching) is of size $\tfrac{n-\vert\mathcal{O}_I\vert}{2}$. Therefore $n-\vert\mathcal{O}_I\vert\geq n-r_1-r_3$, i.e., $\vert\mathcal{O}_I\vert\leq r_1+r_3$. Clearly, we can split the odd-length cycles into cycles of length at least 3 and cycles of length 1 as follows: $\vert\mathcal{O}_I\vert=\vert\mathcal{O}_I^{\geq 3}\vert+\vert\mathcal{O}_I^{=1}\vert$. Thus, $r_1-\vert\mathcal{O}_I^{=1}\vert\geq \vert\mathcal{O}_I^{\geq 3}\vert-r_3$, but $0\leq r_1\leq 1$ (otherwise the two agents in $P_1$ components must strictly prefer each other to their best neighbour by complete preference lists) and $0\leq \vert\mathcal{O}_I^{=1}\vert\leq 1$ (otherwise the two agents in fixed points must strictly prefer each other to themselves by complete preference lists). Therefore, finally, $\vert\mathcal{O}_I^{\geq 3}\vert-r_3\leq 1$, i.e., $\vert\mathcal{O}_I^{\geq 3}\vert-1\leq r_3$, as required.
\end{proof}

\begin{corollary}
    Let $I=(A,\succsim)$ be a preference system. If the preferences in $I$ are strict, then the $(\vert\mathcal{O}_I\vert-\vert\mathcal{O}_I^{\geq 3}\vert,\frac{n-\vert\mathcal{O}_I\vert}{2}-\vert\mathcal{O}_I^{\geq 3}\vert,\vert\mathcal{O}_I^{\geq 3}\vert)$-bundle $B$ computed using Algorithm \ref{alg:klmbundle} has at most one more $P_3$ component than any other $(r_1,r_2,r_3)$-bundle of $I$, so $B$ is a an additive 1-approximation of a bundle with a minimum number $m'$ of $P_3$ components. 
\end{corollary}

Unfortunately, this result does not hold for instances with ties in the preference lists due to the arbitrary tie break involved in the algorithm. Here, our approach can be suboptimal when trying to minimise the number of $P_3$ components. First, note that the stability of a given arrangement can be verified efficiently.

\begin{proposition}
\label{prop:verify}
    Let $I=(A,\succsim)$ be a preference system with $n$ agents, let $G=(V,E)$ be a graph, let $A'\subseteq A$ be any subset of the agents, and let $M:A'\rightarrow  V$ be an arrangement. Then we can decide whether $M$ is stable in $O(\vert A'\vert^2)$ time.
\end{proposition}
\begin{proof}
    We simply need to consider every pair of agents that is assigned to a vertex (i.e., $O(\vert A'\vert^2)$ pairs) and verify that they do not strictly prefer each other to their respective best neighbours. Assuming that we keep track of every agent's best neighbour, we can do this comparison in constant time.
\end{proof}

We now highlight the NP-hardness of minimising $P_3$s in the bundle and establish further properties of the output of Algorithm \ref{alg:klmbundle} in the case when preferences may contain ties.

\begin{theorem}
    Let $I=(A,\succsim)$ be a preference system with $n$ agents and let $r_1,r_2,r_3$ be non-negative integers. If a preference system $I$ contains ties in the preference lists, then it is NP-complete to decide whether $I$ admits a $(r_1,r_2,r_3)$-bundle $B$. Specifically, in this setting, it is NP-hard to find a $(r_1,r_2,r_3)$-bundle with minimum $r_3$, and it is NP-hard to approximate the minimum $r_3$ within any multiplicative factor. However, the $(\vert\mathcal{O}_{I'}\vert-\vert\mathcal{O}_{I'}^{\geq 3}\vert,\frac{n-\vert\mathcal{O}_{I'}\vert}{2}-\vert\mathcal{O}_{I'}^{\geq 3}\vert,\vert\mathcal{O}_{I'}^{\geq 3}\vert)$-bundle (for any tie-break instance $I'$ of $I$) computed using Algorithm \ref{alg:klmbundle} gives an additive approximation guarantee of $\left\lfloor\tfrac{n}{3}\right\rfloor$.
\end{theorem}
\begin{proof}
    Recall that it is NP-complete to decide whether an {\sc srt} instance with $n\in 2\mathbb{Z}^+$ agents admits a weakly stable matching \cite{ronn_srt}. Now recall from Section \ref{sec:formal} that problem instances in {\sc srt} consist of a preference system equivalent to the one we adopted throughout this work, except that we also allow odd numbers of agents. We claim that $I$ (with $n$ even) admits a weakly stable matching if and only if $I$ admits a $(0,\tfrac{n}{2},0)$-bundle, in which case NP-hardness of our problem follows immediately.

    Suppose first that $I$ admits a weakly stable matching $M$. By complete preferences and $n$ even, $\vert M\vert=\tfrac{n}{2}$ (otherwise two unmatched agents would strictly prefer each other to being unmatched). Now consider a collection $B$ of $\tfrac{n}{2}$ $P_2$ components and a natural complete arrangement (mapping) $M'$ from $A$ to the components such that matched agents are neighbours in $B$. Then clearly $M'$ is stable by the stability of $M$ (i.e., no two agents strictly prefer each other to their best, in this case their only, neighbour) and so $B$ is a $(0,\tfrac{n}{2},0)$-bundle.

    Now for the other direction, suppose that $I$ admits a $(0,\tfrac{n}{2},0)$-bundle. Then there exists a complete stable arrangement of the agents on $P_2$ components. Consider the matching $M$ corresponding to the arrangement; $M$ is complete, and no two agents in the matching strictly prefer each other to their partners. Thus, $M$ is a stable matching of $I$.

    Membership in NP follows from the fact that a bundle $B$ will contain both a graph $G_B$ and a bijective arrangement $M_B$ of the agents on $G_B$. Bijectivity can be verified in time linear in the number $n$ of agents, and stability can be verified in $O(n^2)$ time due to Proposition \ref{prop:verify}.

    Due to the phenomenon of being able to decide whether $I$ is solvable or not purely based on knowing whether the minimum $r_3$ is 0 or not, we cannot hope for a polynomial-time algorithm with a multiplicative approximation guarantee, unless P=NP. The multiplicative guarantee is unbounded because it might be the case that there exists a $(r_1,r_2,0)$-bundle, but the tie-break causes $\vert\mathcal{O}_{I'}^{\geq 3}\vert>0$. For example, if every preference list consists of a single long tie, then any maximum matching (necessarily a $(r_1,r_2,0)$-bundle) is stable, but the tie break can cause the stable partition of $I'$ to consist of $\left\lfloor\tfrac{n}{3}\right\rfloor$ cycles of length 3.
    
    The additive approximation guarantee of the stable partition bundle follows immediately from the fact that we can compute a bundle with the stated properties using Algorithm \ref{alg:klmbundle}, $r_3=\vert\mathcal{O}_I^{\geq 3}\vert$ (as argued in Theorem \ref{thm:klmbundle}) and there can be at most $\left\lfloor\tfrac{n}{3}\right\rfloor$ $P_3$ components (by completeness of the arrangement and the fact that there are three agents in each $P_3$ component). 
\end{proof}

A key observation to make at this point is that adding edges to a graph that admits a stable arrangement maintains stability. In other words, adding edges can only be an improvement for the agents!

\begin{theorem}
\label{thm:addedges}
    Let $I=(A,\succsim)$ be a preference system and let $G=(V,E)$ be a graph that admits a stable arrangement $M$. Then, for any pair of distinct vertices $v_i,v_j\in V$ such that $\{v_i,v_j\}\notin E$, it must be that $M$ is stable in $G'=(V,E\cup \{\{v_i,v_j\}\})$.
\end{theorem}
\begin{proof}
    Clearly, $M$ does not admit any blocking pairs with respect to $G$. Now suppose that $M$ is not stable in $G'$. Then there must exist a blocking pair $\{a_r,a_s\}$ such that $a_s\succ_r b_{G'}(a_r)$ (where $b_{G'}(a_r)$ is the best neighbour of $a_r$ with respect to $M$ and $G'$) and $a_r\succ_s b_{G'}(a_s)$. However, notice that $b_{G'}(a_r)\succsim_r b_{G}(a_r)$ and $b_{G'}(a_s)\succsim_s b_{G}(a_s)$ because either the best neighbours remain the same (or equally good in the case of a tie) after addition of the extra edge, or the best neighbour is strictly better after addition of the edge. Therefore $\{a_i,a_j\}$ must block $M$ with respect to $G$, a contradiction.    
\end{proof}

We will make extensive use of this result in the next section. To give an initial idea of how this will be useful, recall that we have shown how to find a $(r_1,r_2,r_3)$-bundle efficiently (and that any preference system admits one), so by adding edges, we can arrive at complex structures that admit stable arrangements.

\section{Applications to Agent Arrangement Problems}
\label{sec:applications}

\subsection{Arrangements in General Graphs}

In the beginning of this paper, we introduced the {\sc Arranging Bundles} problem. To argue the NP-completeness of this problem, first note that mapping a bundle onto a graph is dependent on the existence of a subgraph isomorphism. Note that, by Theorem \ref{thm:addedges}, this subgraph need not be induced. For an explanation of what it means to arrange a bundle on a graph, we refer back to Section \ref{sec:contributions}. Throughout, we will use the graph-theoretic notation $r_1P_1\cup r_2P_2$ to denote the disjoint union of $r_1$ $P_1$s and $r_2$ $P_2$s.

\begin{lemma}
    \label{lemma:subgraph}
    Let $G=(V,E)$ be a graph and $B$ be a $(r_1,r_2,r_3)$-bundle, then there exists a complete stable arrangement $M$ of $B$ on $G$ if and only if $G$ contains a subgraph that is isomorphic to the disjoint union $r_1P_1 \cup r_2 P_2\cup r_3P_3$. 
\end{lemma}
\begin{proof}
    One direction is easy: suppose that there exists a complete stable arrangement $M$ of $B$ on $G$. Then the graph $G_B$ corresponding to bundle $B$ satisfies $G_B=r_1P_1\cup r_2P_2 \cup r_3P_3$ by definition. Thus, if neighbours in $B$ are preserved in $M$, a core assumption throughout this paper which we made in Section \ref{sec:contributions}, then the required subgraph must exist trivially.

    Now suppose instead that $G$ contains a subgraph $G'$ that is isomorphic to $r_1P_1\cup r_2P_2 \cup r_3P_3$, i.e., the union of $r_1$ $P_1$ components, $r_2$ $P_2$ components and $r_3$ $P_3$ components. Then we can define $M$ in such a way that it simply maps all $P_1$ components in $B$ to $P_1$ components in $G'$, $P_2$ components in $B$ to $P_2$ components in $G'$, and $P_3$ components in $B$ to $P_3$ components in $G'$. $M$ is necessarily a complete arrangement of $B$ on $G'$, and as $G'$ is a subgraph of $G$, $M$ is a complete arrangement of $B$ on $G$. By the stability of $B$ and the fact that all agents maintain their best neighbours from $B$ in $M$ on $G$, there can be no two agents that prefer each other to their best neighbours in $G$, so $M$ is a complete stable arrangement as required.   
\end{proof}

To establish our complexity result, we will reduce from the {\sc $P_s$-Partition} problem.

\begin{problemBox}
\textsc{$P_s$-Partition} \\[4pt]
\textbf{Input:} A simple graph $G=(V,E)$ with $\vert V\vert=rs$. \\[2pt]
\textbf{Question:} Do there exist $r$ vertex disjoint simple paths of length $s$ in $G$?
\end{problemBox}

\citet{monnot2006} showed that {\sc $P_s$-Partition} is NP-complete for any $s\geq 3$ even if $G$ is bipartite and of maximum degree 3. We will now use this result to establish the complexity of {\sc Arranging Bundles}, which we defined in Section \ref{sec:contributions}.

\begin{theorem}[Restatement of Theorem \ref{thm:arrbundlesNPcompleteintro}]
    {\sc Arranging Bundles} is NP-complete even if $G$ is bipartite and has maximum degree 3.
\end{theorem}
\begin{proof}
    Let $G=(V,E)$ be an instance of {\sc $P_3$-Partition}. We can construct a bundle $B=(G_B,M_B)$ such that $r_1=r_2=0$ and $r_3=\tfrac{\vert V\vert}{3}$, for example by considering a preference system consisting of $\vert V\vert$ agents that all have strict preferences such that the stable partition correponding to the preference system consists of $\tfrac{\vert V\vert}{3}$ 3-cycles (see, for example, \citet{glitzner2025empirics} for an explicit construction of such preference systems). 
    
    Clearly $G$ admits $\tfrac{\vert V\vert}{3}$ vertex disjoint simple paths of length $3$ if and only if $G$ contains a subgraph consisting of $\tfrac{\vert V\vert}{3}$ disjoint $P_3$ components. By Lemma \ref{lemma:subgraph}, $G$ admits a complete stable arrangement $M$ of $B$ on $G$ if and only if $G$ contains a subgraph that is isomorphic to $r_1P_1 \cup r_2 P_2\cup r_3P_3$. Since $r_1=r_2=0$, we can conclude that $G$ is a yes-instance of {\sc $P_3$-Partition} if and only if $(G,B)$ is a yes-instance of {\sc Arranging Bundles}. NP-hardness of {\sc Arranging Bundles} follows immediately. Membership of {\sc Arranging Bundles} in NP follows from the fact that a stable arrangement $M$ of $B$ on $G$ can be verified efficiently, as established in Proposition \ref{prop:verify}.
\end{proof}

However, despite the intractability of {\sc Arranging Bundles} in general, we can give a list of sufficient conditions for the existence of a complete stable arrangement on a given graph that can be verified in polynomial time in the size of the graph. Notice that the time complexity is with respect to the graph, as a graph that admits a complete stable arrangement must have at least as many vertices as agents in the preference system by definition, but could in theory contain arbitrarily many more vertices than agents. Also, recall that we treat bundles as compressions of preference systems, so it is sufficient to deal with bundles rather than with preference systems.

\begin{theorem}
    \label{thm:conditions}
    Let $G=(V,E)$ be a graph and $B$ be a $(r_1,r_2,r_3)$-bundle, then there exists a complete stable arrangement $M$ of $B$ on $G$ if:
    \begin{itemize}
        \item $G$ admits a path (without repeated vertices) of length $r_1+2r_2+3r_3$, or 
        \item there are at least $r_1+r_2+r_3$ connected components of size at least 3, or
        \item $G$ is a tree of height at least 3 (where the root has depth 0) with at least $r_2+r_3$ leaves.
    \end{itemize}
    Each of these independent sufficient conditions can be verified in polynomial time (in the size of the graph).
\end{theorem}
\begin{proof}
    If $G$ admits a path of length $r_1+2r_2+3r_3$ (i.e., the number of agents in our instance $n$), then we can arbitrarily join the path components in our bundle at the endpoints to arrive at a stable path of length $n$, which must be stable when mapped to the isomorphic $P_n$ subgraph in $G$ by the fact that every agent maintains their best neighbour from $B$ in $G$, as established in Theorem \ref{thm:addedges}.  
    
    Similarly, if $G$ admits $r_1+r_2+r_3$ connected components of size at least 3, then we can simply map all connected components in $B$ (i.e., $r_1+r_2+r_3$ paths involving at most 3 agents each) to arbitrary such components in $G$. Stability follows by the same reasoning as before.

    Finally, if $G$ is a tree of height at least 3 with at least $r_2+r_3$ leaves, then $G$ contains at least $r_2+r_3$ non-overlapping paths involving three vertices, plus the root vertex. Given that $r_1\in\{0,1\}$ as established in Theorem \ref{thm:klmbundle}, we can always map a $P_1$ in $B$ to the root vertex in the tree (if such a $P_1$ exists) and all $r_2$ $P_2$s and all $r_3$ $P_3$s to the disjoint paths in the tree. Stability follows from the same reasoning as before.

    All of these conditions can be checked by an efficient exploration of the graph.
\end{proof}

\subsection{Designing Seating Graphs}

Returning to our original motivating example of designing stable seating arrangements, e.g., at conference or wedding dinners, we can apply our framework based on $(r_1,r_2,r_3)$-bundles to find suitable solutions efficiently. Of course, seating graphs can take a variety of forms, depending on the forms and sizes of the tables available, physical constraints on the room in which tables are placed, and assumptions on the communication directions that agents seated at the tables follow. However, we will show that our results apply very generally by constructing stable seating paths that can then be joined arbitrarily (as long as they are derived by adding edges to the collection of paths) to arrive at, for example, rows of tables or round tables, or more graph-theoretic structures such as cliques or grids. Note that if any condition of Theorem \ref{thm:conditions} applies, e.g., if we have a target seating graph that admits a path $P$ of length $n$ (where $n$ is the number of agents looking to sit), then we can easily compute a complete stable arrangement in $O(n^2)$ time by arbitrarily joining together the components of a $(r_1,r_2,r_3)$-bundle $B$.

In Section \ref{sec:mintables} we consider the problem of arranging bundles to minimise the number of tables, subject to a given table size, while in Section \ref{sec:minseats} we consider the problem of arranging bundles to minimise the number of seats, subject to a given number of tables, before we finally give results for the hierarchical optimisation problems stated in the beginning of the paper in Section \ref{sec:minminseats}.

\subsubsection{Minimising Tables}
\label{sec:mintables}

Consider the following optimisation problem of finding a stable arrangement with a minimum number of tables. We specifically leave open the structure of each connected component in the definition here.

\begin{problemBox}
\textsc{$s$-SeatsMinTables} \\[4pt]
\textbf{Input:} A non-empty $(r_1,r_2,r_3)$-bundle $B$ and an integer $1\leq s\leq r_1+2r_2+3r_3$. \\[2pt]
\textbf{Output:} A seating graph $G=(V,E)$ and a complete stable arrangement $M$ of $B$ such that $G$ contains $t$ tables (connected components), each table contains $s$ seats (vertices) and $t$ is minimal.
\end{problemBox}

The nature of this problem is the following: we have a number of items (in the bundle), each of a size in $\{1,2,3\}$, and we have components of a fixed capacity $s$. Our objective is to find an assignment of the items to the components such that every item is assigned to a component, no component exceeds its capacity, and the number of components is minimal. This naturally corresponds to the optimisation problem corresponding to the following restricted bin-packing problem, which we will denote by {\sc $\{1,2,3\}$-BinPacking}.

\begin{problemBox}
\textsc{$\{1,2,3\}$-BinPacking} \\[4pt]
\textbf{Input:} Integers $r_1,r_2,r_3,s$. \\[2pt]
\textbf{Output:} A minimum integer $t$ of bins such that there exist a packing of $r_1$ items of size 1, $r_2$ items of size 2 and $r_3$ items of size 3 to $t$ bins with capacity $s$.
\end{problemBox}

\begin{proposition}
    Let $(B,s)$ be an instance of {\sc $s$-SeatsMinTables} consisting of a $(r_1,r_2,r_3)$-bundle $B$ and a positive integer $s$. Then a solution $(G,M)$ requiring $t_*$ tables is optimal if and only if {\sc $\{1,2,3\}$-BinPacking} instance $(r_1,r_2,r_3,s)$ has an optimal solution $t_*$.
\end{proposition}

Of course, the general bin packing problem is famously (strongly) NP-hard in general \cite{Garey1979}, but it is known that for a fixed number of item sizes, the problem is polynomial-time solvable (although with a doubly-exponential runtime in the number of different sizes) \cite{binpacking}. While our restricted version appears to be simple at first, it is easy to show that common heuristics such as {\tt FFD} (first-fit-decreasing) \cite{Garey1979}, which considers items in decreasing order and places them in the first bin that has sufficient capacity for them or starts a new one if no bin can accommodate them, can fail. Consider $(r_1,r_2,r_3)=(0,4,2)$ and $s=7$, then {\tt FFD} will place two items of size 3 in one bin, three items of size 2 into a second bin, and the remaining item of size 2 into the last bin. This leaves the first and second bins with free capacity 1 and the third bin with free capacity 5. An optimal packing, however, would pack one item of size 3 and two items of size 2 in each of two bins, leaving no free capacity. 

We will now present a fast (and much simpler than the approach in \cite{binpacking}) dynamic-programming algorithm for {\sc$\{1,2,3\}$-BinPacking} (for the special case that $r_1\in \{0,1\}$ as we observed in Theorem \ref{thm:klmbundle}) that will turn out to be strongly polynomial for {\sc $s$-SeatsMinTables} (due to the input encoding, as we will elaborate on later). Algorithm \ref{alg:binpackingdp} works as follows: we have a \emph{state} triple $(r_1, r_2, r_3, t) \in \mathbb{Z}^4$, representing the counts of remaining items of sizes 1, 2, and 3 and the number of bins $t$ already used. We also generate all possible valid bin configurations $C$, where $(c_1,c_2,c_3)\in C$ is a configuration in which the corresponding bin contains $c_1$ items of size 1, $c_2$ items of size 2, and $c_3$ items of size 3. Then, the algorithm performs a breadth-first search through the possible states starting from the initial state $(r_1,r_2,r_3,0)$, and explores all valid bin configurations whose total size is at most $s$. Each bin configuration will decrease at least one type of remaining items, and every configuration that does not decrease any number of remaining items below zero is valid at this point. For each such valid bin configuration, we add a new state to the queue to explore, in which we deduct the number of items that can be contained in this bin and add a new used bin to the state. The first time the state $(0, 0, 0,t)$ is reached (for some $t\geq 1$), in which case the breadth-first search found a path of state transitions that corresponds to a collection of bin configurations that packs every item, the number of $t$ bins used is guaranteed to be minimal. During the algorithm's execution, we also keep track of the trace $T$ to later reconstruct the bins.

\begin{algorithm}[!htb]
\renewcommand{\algorithmicrequire}{\textbf{Input:}}
\renewcommand{\algorithmicensure}{\textbf{Output:}}

\begin{algorithmic}[1]
\Require{$r_1,r_2,r_3$ : number of items of sizes $1$, $2$, and $3$ respectively; $s$ : bin capacity}
\Ensure{$t$ : number of bins; $T$ : trace to reconstruct configurations}

\State $Q \gets \{(r_1,r_2,r_3, 0)\}$ \Comment{Initialise Queue}
\State $V \gets \emptyset$ \Comment{Initialise set of visited configurations}
\State $T\gets (\mathbb{Z}^3\rightarrow (\mathbb{Z}^3,\mathbb{Z}^3))$ \Comment{Initialise an empty mapping for the reconstruction of bins}

\State$C \gets \{(c_1, c_2, c_3) \mid (c_1, c_2, c_3 \geq 0)\land (c_1\leq 1)\land (c_1 + 2c_2 + 3c_3 \leq s)\land (c_1 + c_2 + c_3 > 0)\}$ \Comment{Compute all valid bin configurations}

\While{$Q$ is not empty}
    \State $(r_1,r_2,r_3, t) \gets Q$.pop()
    \If{$(r_1,r_2,r_3) = (0, 0, 0)$} \Comment{We have arrived at a solution}
        \State \Return{$t,T$}
    \EndIf
    \If{$(r_1,r_2,r_3) \in V$} \Comment{This state has already been explored so it can be skipped}
        \State \textbf{continue}
    \EndIf
    \State $V$.add$((r_1,r_2,r_3))$
    \For{$(c_1, c_2, c_3) \in C$}
        \If{$r_1 \geq c_1$ and $r_2 \geq c_2$ and $r_3 \geq c_3$}
            \State $r_1' \gets r_1 - c_1$
            \State $r_2' \gets r_2 - c_2$
            \State $r_3' \gets r_3 - c_3$
            \If{$(r_1', r_2', r_3') \notin V$}
                \State $T[(r_1', r_2', r_3')] \gets ((r_1,r_2,r_3), (c_1, c_2, c_3))$ \Comment{Track state and bin configurations}
                \State $Q$.enqueue$((r_1', r_2', r_3', t+1))$
            \EndIf
        \EndIf
    \EndFor
\EndWhile

\end{algorithmic}
\caption{{\tt DP} Dynamic programming algorithm for {\sc$\{1,2,3\}$-BinPacking}}
\label{alg:binpackingdp}
\end{algorithm}

\begin{lemma}
\label{lemma:binpackingdb}
    Let $I=(r_1,r_2,r_3,s)$ be a {\sc$\{1,2,3\}$-BinPacking} instance where $r_1\in\{0,1\}$ and let $t_*$ be the number of bins returned by Algorithm \ref{alg:binpackingdp}. Then $t_*$ is an optimal solution to $I$ and Algorithm \ref{alg:binpackingdp} runs in $O(s^2n^2)$ time.
\end{lemma}
\begin{proof}
    Each transition corresponds to adding a new bin configuration and removing the respective item sizes from the currently remaining items state. We also explore all possible such transitions from every state that we consider, and we consider states in a breadth-first search manner. Thus, for any non-initial reachable state $(r_1, r_2, r_3,t)$, the minimal number of bins to reach it from $(r_1,r_2,r_3,0)$ is equal to 1 plus the minimal number of bins needed to reach $(r_1-c_1, r_2-c_2, r_3-c_3,t-1)$ using some valid bin configuration $(c_1,c_2,c_3)$. Given this optimal substructure of the problem, the number of bins for a state is computed by minimising over all valid one-bin transitions from predecessor states. Here, the minimisation is implicit: at every step the algorithm considers every possible valid bin, and it terminates as soon as it encounters a sequence of bins that leaves no item unpacked. Of course, in general, the algorithm explores many suboptimal packings along the way. Now, since each transition reduces the total number of remaining items and the number of states is finite, the dynamic program must terminate.

    Note that $C$ contains $O(c_1^*\times c_2^*\times c_3^*)$ possible configurations, where $c_1^*$, $c_1^*$ and $c_1^*$ maximise the set of equations from line 4 of the algorithm, i.e., $c_1^*\leq 1$, $c_1^*+2c_2^*+3c_3^*\leq s$ and $c_1^*+c_2^*+c_3^*>0$. Hence, $c_1^*= 1$, $c_2^*= \tfrac{s}{2}$ and $c_3^*= \tfrac{s}{3}$, therefore, asymptotically, $C$ contains $O(s^2)$ configurations. Similarly, we have $O(r_1\times r_2\times r_3)$ possible states, but given that $r_1\in\{0,1\}$, $r_2\leq\tfrac{n}{2}$ and $r_3\leq\tfrac{n}{3}$ (by $n=r_1+2r_2+3r_3$), this is bounded by $O(n^2)$. Clearly, the dynamic program explores every state at most once in the while loop (due to the condition in line 10), and for every state we explore all possible valid bin combinations (due to the for loop in line 14), yielding an overall time complexity of $O(s^2n^2)$ as required.
\end{proof}

Now for {\sc $s$-SeatsMinTables}, we actually need the exact bin configurations used in an optimal solution, not just the number of bins. This is where the trace $T$ comes in that Algorithm \ref{alg:binpackingdp} creates and returns. Algorithm \ref{alg:reconstructbins} reconstructs the bin configurations as follows: an initial empty collection of bin configurations is constructed, and we start with the final state of remaining items $(0,0,0)$. To see which bin configuration was used to reach the final state, we query $T$, which in turn returns the previous state and the respective bin configuration, which is added to our collection of bin configurations used. This process continues until we arrive at the initial state of the dynamic program containing $r_1,r_2,r_3$, in which case we have retraced a shortest path of bin configurations.

\begin{algorithm}[!htb]
{\small
\renewcommand{\algorithmicrequire}{\textbf{Input:}}
\renewcommand{\algorithmicensure}{\textbf{Output:}}

\begin{algorithmic}[1]
\Require{$T$ : Trace; $(r_1,r_2,r_3)$ : initial state}
\Ensure{$P$ : packing, i.e., all used bin configurations}

\State $P \gets [\;]$ \Comment{Keep track of all bin configurations used}
\State $q \gets (0, 0, 0)$ \Comment{Start from the final state}

\While{$q \neq (r_1,r_2,r_3)$}
    \State $(q_{\text{prev}}, bin) \gets T[q]$
    \State $P$.add$(bin)$
    \State $q \gets q_{\text{prev}}$
\EndWhile

\State \Return{$P$}
\end{algorithmic}
\caption{{\tt Reconstruct} Reconstructs bin configurations from dynamic-program trace}
\label{alg:reconstructbins}}
\end{algorithm}

\begin{lemma}
    \label{lemma:reconstruct}
    Algorithm \ref{alg:reconstructbins} reconstructs the bin configurations of a minimum number of bins $t_*$ in $O(t_*)$ time.
\end{lemma}
\begin{proof}
    By construction, the collection of bin configurations returned can only contain valid bin configurations when Algorithm \ref{alg:reconstructbins} is executed on the output of the dynamic program (because the dynamic program only constructs valid bin configurations). Furthermore, given that we trace the exact sequence of states and their transitions, the set $P$ returned by Algorithm \ref{alg:reconstructbins} contains all necessary and specifically a minimum number of bins (by the correctness of Algorithm \ref{alg:binpackingdp}). The efficiency follows from the fact that the while loop is executed $O(t_*)$ times (because we strictly increase our current state and precisely trace the minimum sequence of transitions), and we assume that querying $T$ and adding to the collection $P$ can be performed in constant time. 
\end{proof}

Now, we can use these two new algorithms to actually compute a solution for the {\sc $s$-SeatsMinTables} problem as we set out to do. Algorithm \ref{alg:mintables} does this: it executes the dynamic program and reconstructs the bin configurations using {\tt DP} and {\tt Reconstruct}, respectively. Then, in lines 3-5, it generates pointers to the $P_1,P_2$ and $P_3$ components in the bundle graph $G_B$ and, in line 6, iterates through all bin configurations of a solution with a minimum number of bins used. For each such configuration $(c_1,c_2,c_3)$, the algorithm joins together $c_1P_1$s with $c_2P_2$s and $c_3P_3$s into one path of length $c_1+2c_2+3c_3$ in $G_B$, while $M_B$ remains unchanged ($M_B$ remains stable due to Theorem \ref{thm:addedges}). If the size of the bin configuration is less than the table size $s$, then, in the for-loop starting in line 29, empty seats (i.e., vertices that no agent is mapped to) are added and connected to the component such that all components have exactly $s$ seats, i.e., each component is a path of length $s$.

\begin{algorithm}[!htb]
{\small
\renewcommand{\algorithmicrequire}{\textbf{Input:}}
\renewcommand{\algorithmicensure}{\textbf{Output:}}

\begin{algorithmic}[1]
\Require{$B=(G_B,M_B)$ : a $(r_1,r_2,r_3)$-bundle, $s$ : number of seats per table}
\Ensure{$G$ : a seating graph}

    \State $t,T \gets$ {\tt DP}$((r_1,r_2,r_3),s)$
    \State $P \gets$ {\tt Reconstruct}$(T,(r_1,r_2,r_3))$

    \State $p_1 \gets [$pointer to $P_1$ in $G_B]$
    \State $p_2 \gets [$pointers to $P_2$ in $G_B]$
    \State $p_3 \gets [$pointers to $P_3$ in $G_B]$
    
    \For{$(c_1,c_2,c_3) \in P$}
        \State $v\gets$ null pointer
        \For{$1\leq i\leq c_1$}
            \State $v\gets p_1$.pop$()[x]$ \Comment{Returns pointer $p_1$ to a $P_1:x$ in $G_B$}
        \EndFor
        \For{$1\leq i\leq c_2$}
            \State $p\gets p_2$.pop() \Comment{Returns pointer $p_2$ to a $P_2:x-y$ in $G_B$}
            \If{$v$ is null}
                \State $v\gets p[x]$
            \Else
                \State $G_B$.addEdge$(\{v,p[x]\})$
                \State $v\gets p[y]$
            \EndIf
        \EndFor
        \For{$1\leq i\leq c_3$}
            \State $p\gets p_3$.pop() \Comment{Returns pointer $p_3$ to a $P_3:x-y-z$ in $G_B$}
            \If{$v$ is null}
                \State $v\gets p[x]$ 
            \Else
                \State $G_B$.addEdge$(\{v,p[x]\})$
                \State $v\gets p[z]$
            \EndIf
        \EndFor
        \For{$1\leq i\leq s-(c_1+2c_2+3c_3)$} \Comment{Fill connected components with empty seats}
            \State $G_B$.addVertex$(x_i)$
            \State $G_B$.addEdge$(\{v,x_i\})$
            \State $v\gets x_i$
        \EndFor
    \EndFor

    \State \Return{$G_B$}
\end{algorithmic}
\caption{{\tt MinTables} Computes a seating graph and arrangement using a minimum number of tables}
\label{alg:mintables}}
\end{algorithm}

Thus, we can finally give the following complexity result for our studied problem.

\begin{theorem}
\label{thm:mintables}
    {\sc $s$-SeatsMinTables} can be solved in $O(s^2n^2)$ time using Algorithm \ref{alg:mintables}, where $n$ is the number of agents. Furthermore, all tables are paths, but any set of edges can be added to transform the tables into other graph structures.
\end{theorem}
\begin{proof}
    By construction, and as established in Lemmas \ref{lemma:binpackingdb}-\ref{lemma:reconstruct}, Algorithm \ref{alg:mintables} uses a minimum number of tables such that each table has size at most $s$. Furthermore, the algorithm connects the $P_i$ components of the given instance bundle $B$ in such a way that paths are created, and the for loop in line 29 ensures that every component (table) is exactly of size $s$ (potentially containing empty seats). Clearly, the arrangement we return is stable because every agent is still connected to their best neighbour in the original (given) stable arrangement $M_B$. The fact that any set of edges can be added is due to Theorem \ref{thm:addedges}. The complexity comes from the complexities $O(s^2n^2)$ and $O(t_*)$ (where $t_*$ is the minimum number of tables) for Algorithms \ref{alg:binpackingdp}-\ref{alg:reconstructbins} (as established in Lemmas \ref{lemma:binpackingdb}-\ref{lemma:reconstruct}), and clearly $t_*\leq n$. Furthermore, in Algorithm \ref{alg:mintables}, the generation of the pointer sets can be implemented in $O(n)$ time, as there are at most $n$ vertices in $G_B$, so it suffices to scan $G_B$ once while keeping track of the connected component sizes. Algorithm \ref{alg:mintables} also has an outer for loop with $t_*$ iterations and multiple sequential inner for loops that run for $O(n)$ iterations, adding an additional complexity of $O(t_*n)$, which is strictly less than the complexity of Algorithm \ref{alg:binpackingdp}. Thus, the algorithm runs in $O(s^2n^2)$ overall as stated.
\end{proof}

Notice that the input of {\sc$\{1,2,3\}$-BinPacking} is given in the form of integers $r_1,r_2,r_3,s$, which can be encoded in $O(\log(r_1+r_2+r_3+s))$ space (in binary). Thus, the upper bound on the complexity of Algorithm \ref{alg:binpackingdp} is exponential in the input size. However, because we assume in the definition of {\sc $s$-SeatsMinTables} that we are given the bundle $B$ explicitly, and $B$ contains the graph $G_B$ of $r_1+r_2+r_3$ components, our input is encoded in $O(r_1+r_2+r_3+\log(s))$ space, i.e., linear in the number of agents $n=\Omega(r_2+r_3)$ (recall that $r_1\in \{0,1\}$ by Theorem \ref{thm:klmbundle}), rendering our algorithm strongly polynomial in the input size.

\subsubsection{Minimising Seats}
\label{sec:minseats}

We now consider a similar problem, but with the following change: instead of minimising the number of tables subject to a given number of seats at each table, we now aim to minimise the number of seats at each table given a number of tables.

\begin{problemBox}
\textsc{$t$-TablesMinSeats} \\[4pt]
\textbf{Input:} A non-empty $(r_1,r_2,r_3)$-bundle $B$ and an integer $1\leq t\leq r_1+2r_2+3r_3$. \\[2pt]
\textbf{Output:} A seating graph $G=(V,E)$ and a complete stable arrangement $M$ such that $G$ contains $t$ tables (connected components), each table contains $s$ seats (vertices) and $s$ is minimal.
\end{problemBox}

It turns out that this problem can be solved easily using Algorithm \ref{alg:stablesminseats}: starting from an initial table size of $r$ (where $r=\max\{i:1\leq i\leq 3 \land r_i>0\}$ in the graph $G_B$ of instance bundle $B$), we compute the minimum number of tables $t_*$ that are required in a seating arrangement in which each table has $r$ seats using Algorithm \ref{alg:mintables}. If the number of tables required is at most our upper bound on the allowed number of tables $t$ given in our problem instance, then we can add $t-t_*$ tables (if $t_*<t$) as paths on new vertices, and return the seating graph and seating arrangement. Otherwise, we increase the table size to $r+1$ and try again. This continues until we have found a seating graph that satisfies our constraints. Again, notice that the stable arrangement $M_B$ of the instance bundle $B$ remains stable as we only add edges, as established in Theorem \ref{thm:addedges}.

\begin{algorithm}[!htb]
{\small
\renewcommand{\algorithmicrequire}{\textbf{Input:}}
\renewcommand{\algorithmicensure}{\textbf{Output:}}

    \begin{algorithmic}[1]

    \Require{$B=(G_B,M_B)$ : a $(r_1,r_2,r_3)$-bundle; $t$ : an integer}
    \Ensure{$G$ : a seating graph}

    \State $r\gets \max\{i:1\leq i\leq 3 \land r_i>0\}$
    \For{$r\leq s\leq r_1+2r_2+3r_3$}
        \State $G\gets$ {\tt MinTables$(s)$}
        \State $p\gets\#$ connected components in $G$ 
        \If{$p\leq t$}
            \For{$1\leq i\leq t-p$}
                \State $G$.add$(P_s)$ \Comment{Adds a path $P_s$ on $s$ new vertices to $G$}
            \EndFor
            \State\Return{$G$}
        \EndIf
    \EndFor
    
    \end{algorithmic}
    \caption{Constructs a seating graph and arrangement with a minimum number of seats}
    \label{alg:stablesminseats}}
\end{algorithm}

\begin{theorem}
\label{thm:minseats}
    {\sc $t$-TablesMinSeats} can be solved in $O(s_{*}^3n^2)$ time using Algorithm \ref{alg:stablesminseats}, where $n$ is the number of agents and $s_{*}$ is the minimum number of seats. Furthermore, all tables are paths, but any set of edges can be added to transform the tables into other graph structures.
\end{theorem}
\begin{proof}
    Clearly, the seating graph and seating arrangement returned by the algorithm satisfy our constraints (i.e., stability and number of tables used) due to the correctness of Algorithm \ref{alg:mintables} established in Theorem \ref{thm:mintables}. Furthermore, the number of seats $s_*$ at each table is minimal because if there existed a valid seating graph with tables of size smaller than $s_*$, then we would have found it at an earlier iteration. We must start the search with $s\geq r$ because if there is a $P_3$ in our instance bundle graph $G_B$ then it cannot be mapped to a component with fewer than 3 vertices, and similarly for $P_2$. Because we assume a non-empty bundle, we must start with $r\geq 1$. Regarding the complexity of the algorithm, we established in Theorem \ref{alg:mintables} that Algorithm \ref{alg:mintables} ({\tt MinTables}) runs in $O(s^2n^2)$ time (where $s$ is the number of seats) and we need to execute this algorithm exactly $s_*$ times (with parameter $s$ increasing from 1 with each iteration, i.e., $s\in[1,2,\dots,s_*]$), giving the stated time complexity. The inner for loop as it is executed only once (due to the subsequent return statement) and runs for $O(t)$ iterations, each simply adding fixed-size components of size $s_*$ to the graph (which takes $O(s_*t)$ time). However, we assume in the problem definition that $t\leq r_1+2r_2+3r_3=n$, i.e., the inner for loop takes $O(s_*n)$ time and is executed exactly once, requiring strictly less time than our stated complexity $O(s_{*}^3n^2)$.

    The path property of our solution was established in Theorem \ref{alg:mintables}, and the inner for loop of Algorithm \ref{alg:stablesminseats} only ever adds additional paths on new vertices. The fact that any set of edges can be added without impacting the stability of $M_B$ on $G$ again follows from Theorem \ref{thm:addedges}.
\end{proof}

\subsubsection{Min-Min Optimisations}
\label{sec:minminseats}

Recall that we already defined {\sc MinTablesMinSeats} and {\sc MinSeatsMinTables} to denote the problems of finding a seating graph with a minimum number of tables, and subject to that a minimal (fixed) number of seats on each table (and vice versa for {\sc MinSeatsMinTables}). The following result shows that our algorithms can be modified to solve these problems too.

\begin{theorem}[Restatement of Theorem \ref{thm:mintablesseatsintro}]
\label{thm:seatingalgos}
    {\sc MinTablesMinSeats} and {\sc MinSeatsMinTables} can be solved in $O((n^2+{s_*}^3)n^2)$ and $O({s_*}^3n^2)$ time, respectively, where $n$ is the number of agents and $s_*$ is the minimal number of seats per table. Furthermore, in the solutions, all tables are paths, and any set of edges can be added to transform the tables into other graph structures (e.g., cycles, grids, or cliques). Finally, for any table size of at least 3, there always exists a stable seating arrangement.
\end{theorem}
\begin{proof}
    Clearly $s\leq n$, so for {\sc MinTablesMinSeats} we can first use the algorithm for \textsc{$s$-SeatsMinTables} with $s=n$ to find the minimum number of tables $t_*$ required. In the next step, we can use the algorithm for \textsc{$t$-TablesMinSeats} with $t=t_*$ to get the desired solution. Due to the sequential nature of this method, the time complexity is no worse than the algorithms we designed for \textsc{$s$-SeatsMinTables} and \textsc{$t$-TablesMinSeats}, i.e., $O(n^4)$ and $O({s_*}^3n^2)$, respectively. Thus, the total time complexity is $O((n^2+{s_*}^3)n^2)$ as stated.
        
    A simpler approach works for {\sc MinSeatsMinTables}: we simply run Algorithm \ref{alg:stablesminseats} with $t=n$ (because clearly an optimal solution has $t\leq n$), but do not add additional $t-t_*$ empty tables (for a minimum number of required tables $t_*$), i.e., skipping the inner for loop in line 6 of the algorithm. Thus, the complexity of this problem is no worse than that of Algorithm \ref{alg:stablesminseats}, i.e., $O({s_*}^3n^2)$. Clearly, by Theorems \ref{thm:mintables}-\ref{thm:minseats}, the returned seating graph consists of paths, and any set of edges can be added while maintaining a stable solution by Theorem \ref{thm:addedges}. Note also that for any table size at least 3, we can find a stable seating arrangement due to Proposition \ref{prop:p3}.
\end{proof}

Note that we could easily modify the algorithms to also optimise for other optimality criteria, such as maximising the number of first-choice neighbours in the seating graph, with respect to the seating bundle $B$. However, this would only be locally optimal with respect to the starting bundle $B$, but not globally optimal with respect to the original preference system that led to the bundle $B$. Computing a bundle $B$ that maximises first choice neighbours (while minimising $P_3$ components) is bound to be NP-hard even when preferences are strict due to the intractability of finding a stable matching in {\sc sr} that maximises the number of first choices \cite{CooperPhD}.

\subsection{Forming Teams}
\label{sec:teams}

We consider a collection of teams (or groups or coalitions) of agents to be expressible as a reflexive, symmetric and transitive relation of the agents. In the language of graphs, a team can simply be modeled as a clique: a collection of vertices which can be occupied by team members such that every team member is connected to every other team member (here, the team membership relation is induced by agent connectivity). Hence, in order to derive algorithms that partition agents into teams such that the partition satisfies desirable stability properties, we can, again, leverage the fact established in Theorem \ref{thm:addedges} that we can add any collection of edges to a given graph while keeping any stable arrangement on the original graph stable. More formally, we observe the following.

\begin{proposition}
\label{prop:seattoteams}
    Any seating graph $G$ can be transformed into a team graph $G'$ by adding all edges necessary to turn individual graph components into cliques. This maintains the stability of any stable arrangement $M$ on $G$ for $M$ on $G'$. Furthermore, every team is of size at most the number of vertices contained in the largest connected component of $G$.
\end{proposition}

Now consider the following optimisation problems, which are very similar to the previously defined problems {\sc MinTeamsMinSize} and {\sc MinSizeMinTeams}, but specifically require the returned target graphs to only consist of clique components, whereas the previously studied problems left open the structure of the components making up the target graphs.

\begin{problemBox}
\textsc{$s$-SizeMinTeams} \\[4pt]
\textbf{Input:} A non-empty $(r_1,r_2,r_3)$-bundle $B$ and an integer $s\geq 1$. \\[2pt]
\textbf{Output:} A graph $G=(V,E)$ and a complete stable arrangement $M$ such that $G$ contains $\leq t$ teams (cliques) of size $\leq s$ and $t$ is minimal.
\end{problemBox}

\begin{problemBox}
\textsc{$t$-TeamsMinSize} \\[4pt]
\textbf{Input:} A non-empty $(r_1,r_2,r_3)$-bundle $B$ and an integer $t\geq 1$. \\[2pt]
\textbf{Output:} A graph $G=(V,E)$ and a complete stable arrangement $M$ such that $G$ contains $\leq t$ teams (cliques) of size $\leq s$ and $s$ is minimal.
\end{problemBox}

Given our previous results, to map between seating graphs and team graphs, all that is required is to turn connected components into cliques, which Algorithm \ref{alg:converttoteam} does. Now, with the results in the previous section, together with Proposition \ref{prop:seattoteams}, we can establish the following tractability results.

\begin{algorithm}[!htb]
\renewcommand{\algorithmicrequire}{\textbf{Input:}}
\renewcommand{\algorithmicensure}{\textbf{Output:}}

    \begin{algorithmic}[1]

    \Require{$G$ : a seating graph ; $M$ : an seating arrangement on $G$}
    \Ensure{$G$ : a team graph; $M$ : a team arrangement }

    \For{connected component $C=(V_C,E_C)$ in $G$}
        \For{vertices $v_i,v_j\in V_C$ such that $i\neq j$ and $\{v_i,v_j\}\notin E_C$}
            \State $E_C$.add$(\{v_i,v_j\})$
        \EndFor
    \EndFor
    \State\Return{$G,M$}
    
    \end{algorithmic}
    \caption{Converts a seating graph to a team arrangement}
    \label{alg:converttoteam}
\end{algorithm}

\begin{lemma}
\label{lemma:teams}
    {\sc $s$-SizeMinTeams} and {\sc $t$-TeamsMinSize} can be solved in $O(s^2n^2)$ and $O(s_{*}^3n^2)$ time, respectively, where $n$ is the number of agents and $s_*$ is the minimum team size when requiring teams of size $t$.
\end{lemma}
\begin{proof}
    Due to Proposition \ref{prop:seattoteams}, we can use Algorithms \ref{alg:mintables}-\ref{alg:stablesminseats} to compute the respective solutions and map to team graphs using Algorithm \ref{alg:converttoteam}. Correctness and efficiency follow from Theorems \ref{thm:mintables}-\ref{thm:minseats}.
\end{proof}

This lets us note the following for min-min optimisations (defined in Section \ref{sec:contributions}).

\begin{theorem}[Restatement of Theorem \ref{thm:minteamssizeintro}]
    {\sc MinTeamsMinSize} and {\sc MinSizeMinTeams} can be solved in $O((n^2+{s_*}^3)n^2)$ and $O({s_*}^3n^2)$ time, respectively, where $n$ is the number of agents and $s_*$ is the minimal team size. Furthermore, whenever the team size is chosen to be at least 3, there always exists a stable team formation.
\end{theorem}
\begin{proof}
    The same arguments as in Theorem \ref{thm:seatingalgos} apply, but using the methods described in Lemma \ref{lemma:teams} for {\sc $s$-SizeMinTeams} and {\sc $t$-TeamsMinSize}.
\end{proof}

\subsection{Best-Stable $b$-Matching}
\label{sec:matching}

We now turn to the problem {\sc Best-Stable $b$-Matching}, which we defined in Section \ref{sec:contributions}. While other efficient algorithms are known for similar kinds of problems \cite{Irving2007,Fleiner08}, our goal is to show the flexibility of our framework and to highlight that this kind of problem can also be solved very efficiently using our techniques. We propose the simple Algorithm \ref{alg:bmatching} to compute a solution to a given instance of {\sc Best-Stable $b$-Matching}: the algorithm starts with an empty set, then considers each component in the $(r_1,r_2,r_3)$-bundle and adds each connected pair of agents in the component to the matching. Then, in the next step, each pair of agents that is not yet matched together is considered and added to the matching if both agents still have free capacity, as all pairs of agents are mutually acceptable by assumption.

\begin{algorithm}[!htb]
\renewcommand{\algorithmicrequire}{\textbf{Input:}}
\renewcommand{\algorithmicensure}{\textbf{Output:}}

    \begin{algorithmic}[1]

    \Require{$B=(G_B,M_B)$ : a bundle; $b$ : a capacity function of the agents}
    \Ensure{$M$ : a $b$-matching}

    \State $M \gets \varnothing$
    
    \For{component $C\in G_B$}
        \If{$C$ is a $P_2$ on $x-y$}
            \State $M$.add$(\{M_B[x],M_B[y]\})$
        \ElsIf{$C$ is a $P_3$ on $x-y-z$}
            \State $M$.add$(\{M_B[x],M_B[y]\})$
            \State $M$.add$(\{M_B[y],M_B[z]\})$
        \EndIf
    \EndFor

    \While{there exist $a_i,a_j$ such that $i\neq j$, $\{a_i,a_j\}\notin M$, $\vert M(a_i)\vert < b_i$ and $\vert M(a_j)\vert < b_j$}
        \State $M$.add$(\{a_i,a_j\})$
    \EndWhile    
    
    \State\Return{$M$}

    \end{algorithmic}
    \caption{Constructs a best-stable $b$-matching}
    \label{alg:bmatching}
\end{algorithm}

The theorem below states the correctness and efficiency of this approach.

\begin{theorem}[Restatement of Theorem \ref{thm:bestbmatchingintro}]
    {\sc Best-Stable $b$-Matching} can be solved in $O(n^2)$ time (where $n$ is the number of agents) using Algorithm \ref{alg:bmatching}.
\end{theorem}
\begin{proof}
    First, notice that $M$ is a $b$-matching by construction: it consists of unordered pairs of agents, and after the for loop, no agent can be contained in more than two pairs. Because we assume in the problem definition that every agent $a_i$ has capacity $b_i\geq 2$ (if $b_i=1$ then the existence of a best-stable $b$-matching is not guaranteed), all capacities are respected at this point in the algorithm. Now notice that one of the conditions in the while loop is that both agents must have free capacity, so adding pairs to the matching in the while loop never makes an agent exceed their capacity.

    Our required stability condition ``no two agents prefer each other to their best partner in $M$'' follows from the fact that in the bundle $B$, no two agents prefer each other over their best neighbour, and every agent matched in $M$ is matched to their best neighbour in $B$.

    The returned matching must be maximal because if there existed an edge that we could add to the matching without exceeding any agent's capacity, then this edge must have already been added to the matching in the while loop. 

    For the complexity, note that $B$ contains $O(n)$ components and each component contains at most 3 agents. Therefore, assuming constant-time addition to sets, the for loop requires $O(n)$ time. The while loop can be implemented in such a way that every pair of agents is considered only once, as they are either added to $M$, in which case they can never be added to $M$ again, or not added to $M$ because one of the conditions fails, in which case the condition must still fail after any other additions to $M$. Thus, the algorithm terminates, and in the while loop, we must consider $O(n^2)$ pairs of agents. Assuming that we keep track of every agent's number of matches in $M$ and that we can test set membership in constant time, the while loop runs in $O(n^2)$ time. The sequential nature of the for and while loops justifies the stated overall time complexity $O(n^2)$.
\end{proof}

\section{Conclusion}
\label{sec:conclusion}

We introduced a versatile framework for agent arrangement problems in which agents have preferences over one another, and where the designer has flexibility in constructing the underlying graph. Within this framework, we provided both positive results and NP-hardness results. Along the way, we uncovered interesting structural and algorithmic connections to classical algorithmic problems such as subgraph isomorphism, disjoint path partitioning, and bin packing, developing results that may be of independent interest. We also explored the tension between local and global optimality, for example, in the context of minimising the necessary resources, such as the number of tables at an event. We gave an algorithm that minimises the number of tables required to arrange everyone in a stable way and showed that it is optimal with regard to a provided bundle, but not necessarily globally optimal with regard to an original preference system -- a trade-off which we made in order to arrive at polynomial-time algorithms.

Most crucially, our paper seeks to challenge the fundamental setup of computationally difficult problems such as {\sc Seat Arrangement} by trying to answer how much tractability can be won by sacrificing some rigidity of the target graph. We imagine that this idea, which is also the underlying premise of works such as those on \emph{near-feasible stable matchings} \cite{glitzner2025unsolvabilitymanytomanynonbipartitestable,nguyen18,cembrano25,gergely25,nearfeasibleaamas}, could help in the design of good algorithms and mechanisms for real-world applications.

With respect to the specific technical contributions of this paper, we believe that the stability concept and tools developed in this work have a range of other applications to problems in multi-agent systems and computational social choice, for example, roommate-room problems \cite{hosseini2025strategyproof,roommateroom} or minimum-regret stable seating arrangements. Furthermore, on the technical side, we conjecture that, when preferences are strict, Algorithm \ref{alg:klmbundle} actually computes a bundle with the minimum number of $P_3$ components, not merely an approximation. Independently, it would also be interesting to consider the implications of incomplete preferences on our framework and note that this impacts the number of possible $P_1$ components in the $(r_1,r_2,r_3)$-bundles, which we showed to be highly restricted in the case of complete preferences (therefore increasing the complexity of Algorithm \ref{alg:binpackingdp}, for example).

\paragraph{Acknowledgements}
The author would like to thank David Manlove for very insightful comments and discussions, the anonymous MATCH-UP reviewers and attendees for helpful feedback and discussions, and acknowledge financial support from a Minerva scholarship from his institution. 

\bibliographystyle{ACM-Reference-Format}
\bibliography{papers}

\end{document}